\documentclass[12pt]{article}
\usepackage{booktabs}
\usepackage{graphics,color}
\usepackage{graphicx}
\usepackage{multirow}
\usepackage{amsmath,amsthm,amssymb,amscd}
\usepackage{amsfonts}
\usepackage{amsmath,amssymb,amsthm,enumerate,epsfig,graphicx}
\usepackage{ifthen,latexsym,syntonly}
\usepackage[sort, comma]{natbib}
\usepackage{verbatim}
\usepackage{soul,color}
\usepackage{tabularx,booktabs}
\usepackage{float}
\usepackage{graphicx}
\usepackage{amsmath,amsthm,amssymb,amscd}
\usepackage{amsfonts}
\usepackage{amsmath,amssymb,amsthm,enumerate,epsfig,graphicx}
\usepackage{ifthen,latexsym,syntonly}
\usepackage{booktabs}
\usepackage{paralist}
\usepackage{color,soul, tabularx}

\usepackage{bm}
\usepackage{algorithm2e}
\theoremstyle{definition}

\newtheorem{theorem}{Theorem}[section]

\newtheorem{lemma}{Lemma}[section]

\usepackage{caption}
\usepackage{multirow}

\usepackage{booktabs}
\usepackage{varwidth}
\newcommand{\ignore}[1]{}

\allowdisplaybreaks
\makeatletter

\newcommand{\Rmnum}[1]{\expandafter\@slowromancap\romannumeral #1@}
\makeatother

\graphicspath{{Figs/}}



\begin{document}
\title{{Objective Bayesian analysis for the generalized
exponential distribution}}
%\author{ Guangying Liu\thanks{Corresponding author. $\textit{Email address}$: liugying@nau.edu.cn}, Ziyan Zhuang, Min Wang \\
%\small{\emph{School of Statistics and Data Sciences, Nanjing Audit University, Nanjing 211815, China}
%}
%}

%\thanks{Research partially supported by National Natural Science Foundation of China (71971118, 11831008, 11971235); National social science foundation project of China (19BTJ035); Jiangsu Province College Science Key Foundation (21KJA110003); Natural Science Foundation of Jiangsu Province of China (BK20221348).}

\author{Aojun Li \ \ and \ \ Keying Ye \ \ and \ \ Min Wang\footnote{Corresponding author. min.wang3@utsa.edu (Min Wang).}\\
{\small{Department of Management Science and Statistics, The University of Texas, San Antonio, TX}}
}

\date{}
\maketitle
\vspace*{-1cm}

\begin{center}
{\small\bf Abstract}
\end{center}
\begin{center}
\begin{minipage}{145mm}
{\small  In this paper, we consider objective Bayesian inference of the generalized exponential distribution using the independence Jeffreys prior and validate the propriety of the posterior distribution under a family of structured priors. We propose an efficient sampling algorithm via the generalized ratio-of-uniforms method to draw samples for making posterior inference. We carry out simulation studies to assess the finite-sample performance of the proposed Bayesian approach. Finally, a real-data application is provided for illustrative purposes.
}

%{{\bf JEL Classification}: C53; C45; C22}
\end{minipage}
\end{center}

%\newpage
\section{Introduction}

A random variable $X$ is said to follow a two-parameter generalized exponential (GE) distribution if its cumulative distribution function (CDF) is given by
\begin{align}
  F(x;\alpha, \lambda) &= \left(1-e^{-\lambda x}\right)^\alpha, \quad x >0, \label{eq:cdf}
\end{align}
where $\alpha > 0$ is a shape parameter and $\lambda > 0$ is a scale parameter. The corresponding probability density function (PDF) has the following form
\begin{align}
  f(x;\alpha, \lambda) &= \alpha\lambda \left(1-e^{-\lambda x}\right)^{\alpha-1}e^{-\lambda x}, \quad x >0. \label{eq:pdf}
\end{align}

The GE distribution was introduced by  \cite{gupta_theory_1999} as an alternative to the gamma and Weibull distributions for analyzing skewed data. Since then, statistical inference for the GE distribution has received much attention from Bayesian and  frequentist perspectives; see, for example, \citet{jaheen_empirical_2004}, \citet{raqab_bayesian_2005}, \citet{gupta_generalized_2007}, and \citet{kundu_generalized_2008}. Nadarajah (2011) provided a comprehensive review of the developments and appealing properties of the GE distribution.

Bayesian analysis of the GE distribution has been primarily conducted under weakly informative priors, such as the gamma distribution. However, in the absence of prior knowledge, it is natural to employ noninformative priors for the unknown parameters, as demonstrated by \citet{moala_bayesian_2012}, \citet{achcar_generalized_2015}, and \citet{naqash_bayesian_2016}. Special attention needs to be paid for the propriety of the posterior distribution under these priors, given that they are usually improper and may result in improper posterior distributions, making statistical inference invalid. For instance, \citet{moala_bayesian_2012} demonstrated that an improper maximal-data-information prior for the GE parameters resulted in an improper posterior distribution.

\ignore{
It is known that within a Bayesian framework, special attention has been paid to specifying appropriate priors for the unknown parameters. In the aforementioned literature, Bayesian analysis of the GE distribution was often conducted under informative priors (e.g., the gamma distribution) for the unknown parameters. However, in an absence of prior knowledge, it seems natural to employ noninformative priors (e.g., Jeffreys prior) for the unknown parameters; see, for example, \citet{moala_bayesian_2012}, \citet{achcar_generalized_2015}, \citet{naqash_bayesian_2016}. It deserves mentioning that noninformative priors are usually improper and could lead to the improper posterior distributions, and thus statistical inference from the corresponding posterior distribution may be invalid. For instance, \citet{moala_bayesian_2012} showed that the improper maximal data information prior for the GE parameters results in an improper posterior distribution.
}

Therefore, this paper revisits the commonly used noninformative prior structure to reflect prior ignorance, the independence Jeffreys prior, for the GE distributional parameters. We provide a theoretical investigation of the propriety of the posterior distribution of $(\alpha,\lambda)$ under the vague prior $\pi_{p}(\alpha,\lambda) = 1/(\alpha^a\lambda^b)$ for $a \in\mathbb{R}$ and $b \in\mathbb{R}$. We show that the resulting joint posterior distributions of the two parameters is proper for any $a \ge 1$ and $b \le 1$. We then propose a generalized ratio-of-uniforms method to draw samples for making posterior inference. Unlike traditional Markov chain Monte Carlo  methods, the proposed algorithm generates independent posterior samples and does not require any burn-in period, since no starting values are required.

\ignore{
This observation motivates us to revisit a commonly used noninformative prior structure to reflect prior ignorance, the independence Jeffreys prior for the GE distributional parameters. To be more specific, we provide a theoretical investigation on the propriety of the posterior distribution of the unknown parameters under the following vague prior for $(\alpha, \lambda)$ given by $\pi_{p}(\alpha,\lambda) = 1/(\alpha^a\lambda^b)$, where $a \in\mathbb{R}$ and $b \in\mathbb{R}$ are hyperparameters. It is shown that the posterior distributions of the two parameters under the family of prior are proper for any $a \ge 1$ and $b \le 1$. Since the joint posterior distribution of the parameters is not a commonly recoganized distribution, we adopt a generalized ratio-of-uniforms method to generate a posterior Monte Carlo sample for making inferences. Unlike other Markov chain Monte Carlo (MCMC) methods, the proposed sampling procedure not only generates independent posterior samples but also does not need any burn-in period, since no starting values are required for the proposed algorithm.
}

The remainder of this paper is organized as follows. In Section  \ref{Section:02}, we derive the independence Jeffreys prior and study the propriety of the posterior distributions under a family of vague priors. In Section \ref{Section:03}, we adopt the generalized ratio-of-uniforms method to draw samples for making posterior inference. In Section \ref{Section:04}, we compare methods for estimating the unknown model parameters through simulations. In Section \ref{sec:app}, we provide a real-data application for illustrative purposes. Finally, we conclude with remarks in Section \ref{sec:con}, and the proof of the main result is deferred to the Appendix.

\ignore{
The remainder of this paper is organized as follows. In Section \ref{Section:02}, we derive the independence Jeffreys prior and study the propriety of the posterior distributions under a family of vague priors. In Section \ref{Section:03}, we adopt the generalized ratio-of-uniforms method for posterior simulations. In Section \ref{Section:04} we carry out simulations to compare methods for estimating the unknown model parameters. In Section \ref{sec:app}, we provide a real-data application for illustrative purposes. Finally, concluding remarks are given in Section \ref{sec:con} with the proof of the main result deferred to the Appendix.
}

\section{Bayesian estimation} \label{Section:02}

In a Bayesian paradigm, one of the most commonly used noninformative priors is the independence Jeffreys prior \citep{jeffreys_theory_1961}, which can be constructed as the product of Jeffreys prior for each unknown parameter, while treating other parameters to be fixed. Recall from \citet{nadarajah_exponentiated_2011} that a simpler form of the Fisher information of the parameters $(\alpha, \lambda)$ can be written as
%$\pi_J(\alpha, \lambda) \propto \sqrt{\mathrm{det}\left(I(\alpha, \lambda)\right)},$
%\begin{align*}
 % \pi_J(\alpha, \lambda) &\propto \sqrt{det I(\alpha, \lambda)},
%\end{align*}
%where $I(\alpha, \lambda)$ is the Fisher information matrix. Its simpler forms were derived by  and can be expressed as
\begin{equation*}
  I(\alpha, \lambda) =
  \begin{bmatrix}
 \frac{n}{\alpha^2}  & \frac{n\{\psi(2)-\psi(\alpha + 1)\}}{\lambda(\alpha-1)}\\
  \frac{n\{\psi(2)-\psi(\alpha + 1)\}}{\lambda(\alpha-1)} & \frac{n}{\lambda^2} + \frac{n\alpha\{\pi^2-6\psi'(\alpha)-12C-12\psi(\alpha)+6C^2+12C\psi(\alpha)+6\psi^2(\alpha)\}}{6\lambda^2(\alpha-2)}
  \end{bmatrix},
\end{equation*}
for $\alpha >0$, $\alpha \neq 1, 2$, where $C$ is the Euler's constant, $\psi$ is denotes the digamma function. \\
%\begin{equation*}
%  I(\alpha, \lambda) =
%  \begin{Bmatrix}
%  E\left(-\frac{\partial^2 \log L}{\partial \alpha^2}\right) & E\left(-\frac{\partial^2 \log L}{\partial \alpha \partial \lambda}\right) \\
%  E\left(-\frac{\partial^2 \log L}{\partial \lambda \partial \alpha}\right) & E\left(-\frac{\partial^2 \log L}{\partial \lambda^2}\right)
%  \end{Bmatrix},
%\end{equation*}
%with
%\begin{alignat*}{1}
%  E\left(-\frac{\partial^2 \log L}{\partial \alpha^2}\right) &= \frac{n}{\alpha^2}, \\
%  E\left(-\frac{\partial^2 \log L}{\partial \alpha \partial \lambda}\right) &= \frac{n\{\psi(2)-\psi(\alpha + 1)\}}{\lambda(\alpha-1)}, \\
%  E\left(-\frac{\partial^2 \log L}{\partial \lambda^2}\right) &= \frac{n}{\lambda^2} + \frac{n\alpha\{\pi^2-6\psi'(\alpha)-12C-12\psi(\alpha)+6C^2+12C\psi(\alpha)+6\psi^2(\alpha)\}}{6\lambda^2(\alpha-2)}.
%
%\end{alignat*}
The independence Jeffreys prior for $(\alpha, \lambda)$ is given by
\begin{align}
  \pi_{IJ}(\alpha,\lambda) \propto \pi_{IJ}(\alpha)\pi_{IJ}(\lambda) \propto \frac{1}{\alpha\lambda}. \label{eq:ij}
\end{align}
In this paper, we consider  the propriety of the posterior distribution under a family of vague priors for $(\alpha,\lambda)$ given by
%By adding hyperparameter parameters $(a, b)$ to \eqref{eq:ij}, we proposed
\begin{align}
  \pi_{p}(\alpha,\lambda) \propto \frac{1}{\alpha^a\lambda^b}, \quad a, b \in\mathbb{R}, \label{eq:pp}
\end{align}
which can be applied in a more general setting.  The resulting joint posterior density function of $\alpha$ and $\lambda$ under the prior in (\ref{eq:pp}) is given by
\begin{equation} \label{joint:poster}
   \pi(\alpha,\lambda \mid \bm X) \propto \alpha^{n-a}\lambda^{n-b} \prod_{i=1}^{n}\left(1-e^{-\lambda x_i}\right)^{\alpha-1}e^{-\lambda \sum_{i=1}^{n}x_i}.
\end{equation}
In what follows, we show that the prior in (\ref{eq:pp}) results in a proper posterior distribution under mild conditions and defer its proof to the appendix.
\begin{theorem}
\label{theorem_proper}
For sample size $n>a-1$, the joint posterior distribution of ($\alpha, \lambda$) under the prior $\pi_{p}(\alpha,\lambda)= 1/(\alpha^a\lambda^b)$ is proper if $a \ge 1$ and $b \le 1$.
\end{theorem}
Note that the independence Jeffreys prior is one of the priors in Theorem \ref{theorem_proper} having proper posterior with $a=1$ and $b=1$.

\section{The generalized ratio-of-uniforms method} \label{Section:03}
We observe from the joint posterior distribution in (\ref{joint:poster}) that the conditional posterior density of $\alpha$ given $\lambda$ and the data is
\begin{align}
  \pi(\alpha \mid \lambda, \bm X) \sim \text{Gamma}\left(n-a+1,-\sum_{i=1}^{n}\ln\left(1-e^{-\lambda x_i}\right)\right), \label{eq:mpost_alpha}
\end{align}
and that the marginal posterior density function of $\lambda$ given the data is
\begin{align}
  \pi(\lambda \mid \bm X)
  &= \lambda^{n-b} e^{-\lambda \sum_{i=1}^{n}x_i} \prod_{i=1}^{n}\left(1-e^{-\lambda x_i}\right)^{-1} \frac{\Gamma(n-a+1)}{\left\{-\sum_{i=1}^{n}\ln\left(1-e^{-\lambda x_i}\right)\right\}^{n-a+1}}. \label{eq:mpost_lambda}
\end{align}

Note that sampling from \eqref{eq:mpost_alpha} can be easily performed in R software, whereas sampling from \eqref{eq:mpost_lambda} may be challenging as it is not tractable analytically. We adopt the ratio-of-uniforms method \citep{wakefield_efficient_1991} to generate independent posterior samples for making statistical inference as follows. We suppose that a pair of random variable $(u,v)$ is uniformly distributed over the region
\begin{align*}
  C(r)=\left[ (u,v): 0 < u \le \left\{\pi\left(\lambda \mid \bm X\right)\right\}^{\frac{1}{r+1}}\right],
\end{align*}
where $r \ge 0$ is a constant, $\pi\left(\lambda \mid \bm X\right)$ is given by \eqref{eq:mpost_lambda}, and $\lambda = v/{u^r}$ has probability density function $\pi\left(\lambda \mid \bm X\right)/\int\pi\left(\lambda \mid \bm X\right)\,d\lambda$. To simulate random samples within $C(r)$, we develop a bounding rectangle $\left[0, a(r)\right] \times \left[ b^-(r), b^+(r)\right]$, with the conditions that the following endpoints of the intervals are bounded, such that
\begin{align*}
  a(r)\!=\!\underset{\lambda>0}\sup\left[\{\pi\left(\lambda\mid \bm X\right)\}^{\frac{1}{r+1}}\right],\quad
  b^-(r)\!=\!\underset{\lambda>0}\inf\left[\lambda\{\pi\left(\lambda \mid \bm X\right)\}^{\frac{r}{r+1}}\right],\quad
  b^+(r)\!=\!\underset{\lambda>0}\sup\left[\lambda\{\pi\left(\lambda \mid \bm X\right)\}^{\frac{r}{r+1}}\right].
\end{align*}
When $\lambda\rightarrow\infty$, we can use L'Hospital's rule to show that $\pi\left(\lambda \mid \bm X\right)$ is bounded and that $\lambda\{\pi\left(\lambda \mid \bm X\right)\}^{\frac{r}{r+1}} = \{\lambda^{\frac{r+1}{r}} \pi\left(\lambda \mid \bm X\right)\}^{\frac{r}{r+1}}$ is bounded. When  $\lambda\rightarrow 0$, we note that $\pi\left(\lambda \mid \bm X\right)$ tends to infinity (e.g., $b=1$). To avoid this problem, we consider a transformation of variable, $Z = g(\lambda) = \ln{\lambda}$, such that $p\left(z \mid \bm X\right)= \pi\left\{g^{-1}(z) \mid \bm X\right\}\left|\frac{d}{dz}g^{-1}(z)\right| = \pi\left(e^z \mid \bm X\right)e^z$ is bounded at $\lambda\rightarrow 0$. Hence, when $\lambda\rightarrow 0$, we first simulate random samples $\{z^{(1)}, \dots, z^{(M)}\}$ from $p\left(z \mid \bm X\right)$ and then convert into samples from $\pi\left(\lambda \mid \bm X\right)$ by inverting the transformation (i.e., $\{\lambda^{(1)}, \dots, \lambda^{(M)}\} = \{g^{-1}(z^{(1)}), \dots, g^{-1}(z^{(M)})\}$).

%When $\lambda\rightarrow\infty$, the first two terms in \eqref{eq:mpost_lambda2} are bounded using L'Hospital's rule. The rest of the terms are bounded near infinity as shown in the proof. Hence, we observe that $\pi\left(\lambda \mid \bm X\right)$ is bounded at $\lambda\rightarrow\infty$. Furthermore, it can be shown that  $\lambda\{\pi\left(\lambda \mid \bm X\right)\}^{\frac{r}{r+1}} = \{\lambda^{\frac{r+1}{r}} \pi\left(\lambda \mid \bm X\right)\}^{\frac{r}{r+1}}$ is bounded at $\lambda\rightarrow\infty$, as the first two terms in \eqref{eq:mpost_lambda2} become $\lambda^{n-b+\frac{r+1}{r}} e^{-\lambda \delta}$, which is still bounded using L'Hospital's rule.
%
%However, when $\lambda\rightarrow 0$, $\pi\left(\lambda \mid \bm X\right)$ could approach infinity (e.g., $b=1$). To avoid this problem, we consider a transformation of variable, $Z = g(\lambda) = \ln{\lambda}$, such that $p\left(z \mid \bm X\right)= \pi\left\{g^{-1}(z) \mid \bm X\right\}\left|\frac{d}{dz}g^{-1}(z)\right| = \pi\left(e^z \mid \bm X\right)e^z$ is bounded at $\lambda\rightarrow 0$. The proof of the boundedness is deferred to the appendix. Hence, when $\lambda\rightarrow 0$, random samples $\{z^{(1)}, \dots, z^{(M)}\}$ could be simulated from $p\left(z \mid \bm X\right)$, then convert into samples from $\pi\left(\lambda \mid \bm X\right)$ by inverting the transformation (i.e., $\{\lambda^{(1)}, \dots, \lambda^{(M)}\} = \{g^{-1}(z^{(1)}), \dots, g^{-1}(z^{(M)})\}$).

As a result of the above, $b^-(r) = 0$, $a(r)$ and $b^+(r)$ are either bounded or bounded via a transformation. Thus, we summarize the generalized ratio-of-uniforms method for posterior sampling in Algorithm \ref{alg:rou}. %The rest of the proof follows in a similar way.\\
\RestyleAlgo{ruled}
\begin{algorithm}
\SetKwInOut{Input}{Input}
\SetKwInOut{Output}{Output}
\ResetInOut{Output}
\caption{The generalized ratio-of-uniforms method}\label{alg:rou}
\KwData{Read data;}
\Input{Set $r$ and the number of sample $M$ to be generated;}
\Output{A sequence of samples $\left(\alpha^{(1)}, \lambda^{(1)}\right), \cdots, \left(\alpha^{(M)}, \lambda^{(M)}\right)$.}
 Compute $a(r)=\underset{\lambda>0}\sup\left[\{\pi\left(\lambda \mid \bm X\right)\}^{\frac{1}{r+1}}\right]$ and $b^+(r)=\underset{\lambda>0}\sup\left[\lambda\{\pi\left(\lambda \mid \bm X\right)\}^{\frac{r}{r+1}}\right]$\;
 \For{$k \gets 1$  \KwTo $M$}{
  Simulate $U \sim \text{Uniform}(0, a(r))$, $V \sim \text{Uniform}(0, b^+(r))$, and compute $\rho = V/U^r$\;
  \While{$U>\{\pi\left(\rho \mid \bm X\right)\}^{1/(r+1)}$}{
   $U \sim \text{Uniform}(0, a(r))$\;
   $V \sim \text{Uniform}(0, b^+(r))$\;
   $\rho = V/U^r$\;
   }{
   $\lambda^{(k)} = \rho$\;
   Simulate $\alpha^{(k)} \sim \text{Gamma}\left(n-a+1,-\sum_{i=1}^{n}\ln\left(1-e^{-\lambda^{(k)} x_i}\right)\right)$.
  }
 }
\end{algorithm}

\section{Simulation studies} \label{Section:04}

We carry out simulations to compare the finite sample performance of Bayesian estimation and maximum likelihood estimation (MLE) in estimating the parameters of the GE distribution. We generate the GE random variables of size $n=(10,15,\cdots,100)$ with various combinations of $\alpha = (0.5, 1, 2)$ and $\lambda = (0.5, 1, 2)$, and for each combination of $(n, \alpha, \lambda)$, we run $N = 200$ replications. To reflect the absence of prior knowledge, we use the independence Jeffreys prior ($a=1$, $b=1$) for Bayesian inference and employ Algorithm 1 to draw $M=10,000$ samples for making the posterior inference. We obtain the MLEs of the parameters by maximizing the profile log-likelihood function; see, for example, \cite{gupta_generalized_2002}.

To compare the performance of the two estimators, we consider the scaled bias (SBias) and the scaled root mean squared error (SRMSE) given by
\begin{align*}
\mathrm{SBias}=\frac{E(\hat{\theta})-\theta}{\theta} \quad \text{and} \quad
\mathrm{SRMSE}=\frac{\sqrt{\text{MSE}(\hat{\theta})}}{\theta},
\end{align*}
where $\hat{\theta} = \sum_{n=1}^{N} \hat{\theta}^{(j)}/N$ and $\hat{\theta}^{(j)}$ is an estimate of $\theta$ for the $j$th replication, $j = 1, \cdots, N$. Numerical results with $\lambda = 1$ are depicted in Figures \ref{fig:Bayes_MLE_GE_alpha_lambda1_pmedian_plot} and \ref{fig:Bayes_MLE_GE_lambda_lambda1_pmedian_plot}. Several conclusions can be drawn as follows.
\begin{itemize}
    \item[(i)] We observe that, for small sample sizes (e.g., $n < 30$), the Bayesian approach has better SBias and SRMSE results (closer to zero) as compared to the MLE method.
    \item[(ii)] For large sample sizes (e.g., $n \ge 30$), the SBias and SRMSE for both approaches converge to the same value as $n$ increases.
    \item[(iii)] The Bayesian approach is recommended for making inferences on the parameters of the GE distribution, especially when sample sizes are small.
\end{itemize}

\begin{figure}[htbp!]
\centering
\caption{SBias and SRMSE for $\alpha$ using the Bayes (solid line) and MLE (dashed line) estimators when {$a = 1$, $b = 1$, $\alpha = (0.5, 1, 2)$, $\lambda = 1$}}
\label{fig:Bayes_MLE_GE_alpha_lambda1_pmedian_plot}
\includegraphics[width=1\linewidth, height = 10cm]{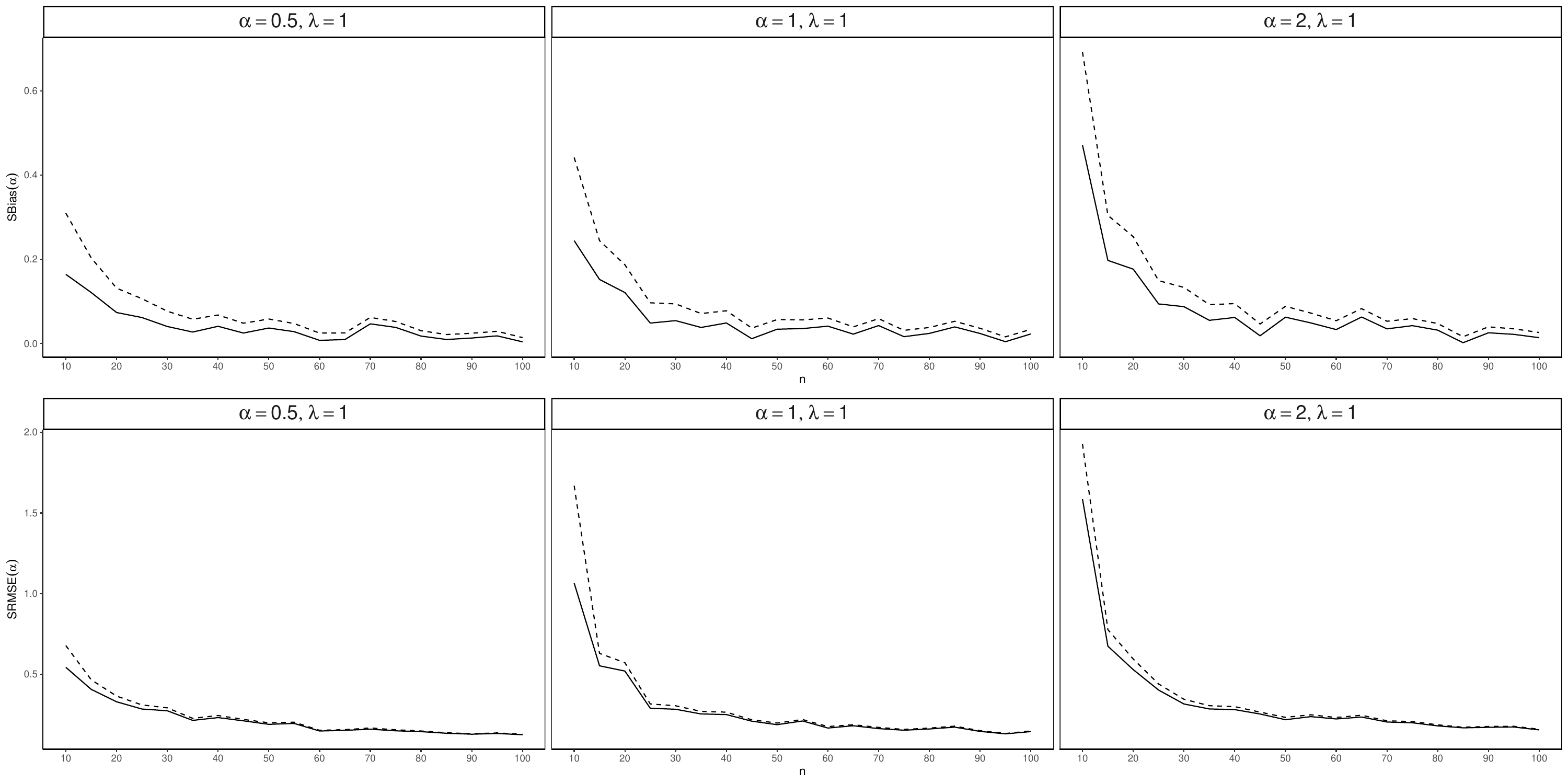}
\end{figure}

\begin{figure}[htbp!]
\centering
\caption{SBias and SRMSE for $\lambda$ using the Bayes (solid line) and MLE (dashed line) estimators when {$a = 1$, $b = 1$, $\alpha = (0.5, 1, 2)$, $\lambda = 1$}}
\label{fig:Bayes_MLE_GE_lambda_lambda1_pmedian_plot}
\includegraphics[width=1\linewidth, height = 10cm]{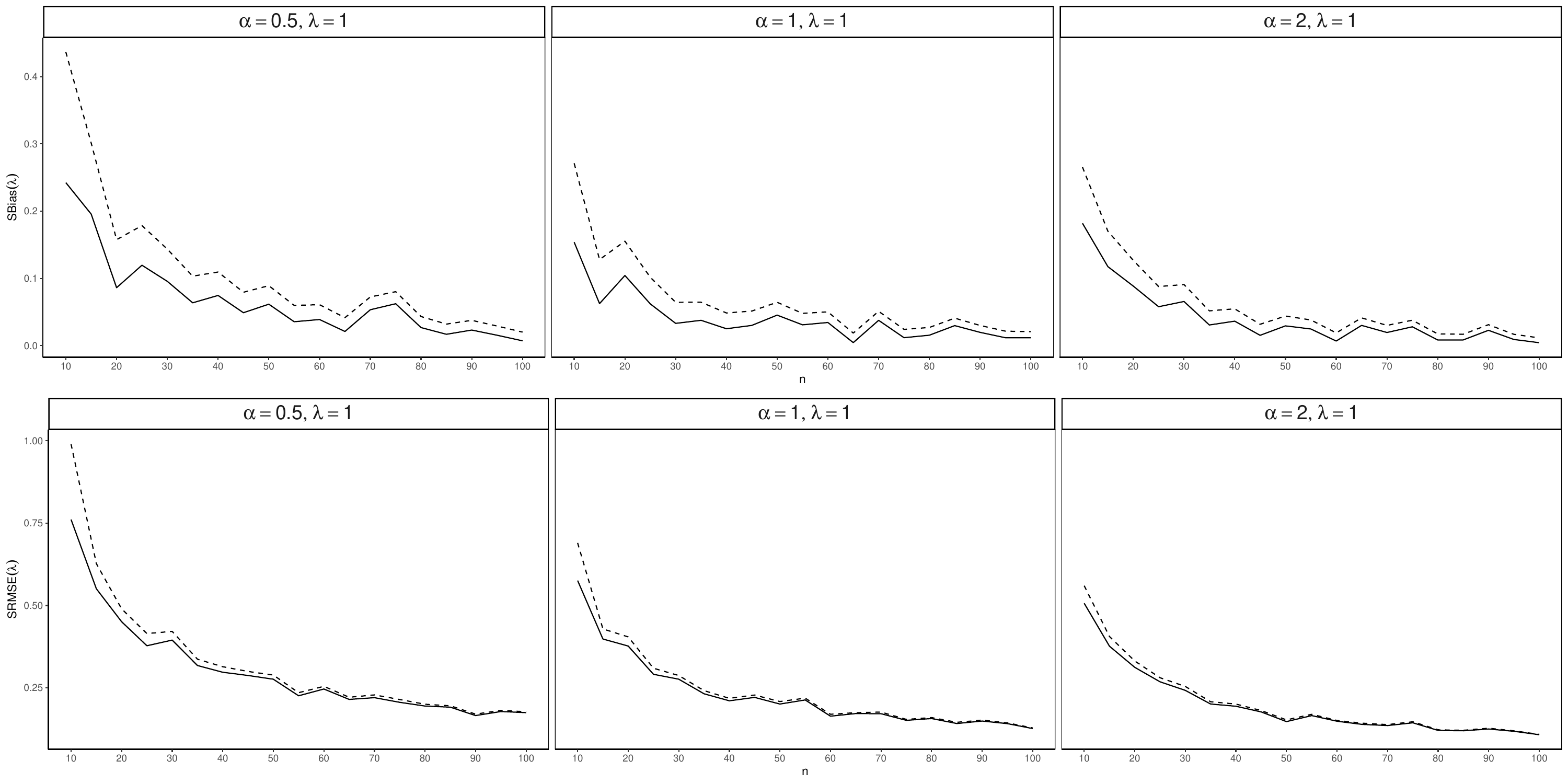}
\end{figure}

%\textcolor{red}{We assess the impact of the hyperparameters in (\ref{eq:pp}) with  $a = (1, 1.5, 2)$ and $b = (-1, 0.5, 1)$. Figure \ref{fig:Bayes_MLE_SBias_GE_pmedian_alpha1_lambda1} shows the SBias and SRMSE when $\alpha = 1$ and $\lambda = 1$. Based on the results, we conclude that Bayesian approach is getting better as compared to the MLE when $a$ and/or $b$ increases. Also, from Theorem \ref{theorem_proper}, the posterior distribution is proper when $a \ge 1$ and $b \le 1$. We observe that the $b = 1$ is the optimal value for prior distributions. Same conclusions were observed for $(\alpha = 0.5, \lambda = 1)$ and $(\alpha = 2, \lambda = 1)$. However, more work with mathematical proof is needed to corroborate this finding.}

%\begin{figure}[ht!]
%\centering
%\caption{SBias for Bayes approach (solid line) and MLE (dashed line) \linebreak %\centering{$a = (1, 1.5, 2)$, $b = (-1, 0.5, 1)$, $\alpha = 1$, $\lambda = 1$}}
%\label{fig:Bayes_MLE_SBias_GE_pmedian_alpha1_lambda1}
%\includegraphics[width=1\linewidth, height = 3.3in]{figures/Bayes_MLE_SBias_GE_pmedian_alpha1_lambda1.pdf}
%\end{figure}

% Also for small sample sizes, when $\lambda$ is fixed, the SBias and SRMSE of $\hat{\alpha}$ with both approaches become larger as $\alpha$ increases.
% When $\lambda =1$, the SBias and SRMSE increase with alpha for estimation of $\alpha$, decrease with alpha for estimation of $\lambda$.

\section{A real-data application}
\label{sec:app}
We consider a real data set from \citet{lawless_statistical_1982}  to illustrate the practical use of the proposed approach. The deep groove ball bearings dataset is a collection of endurance test results for 23 ball bearings, measuring the number of million revolutions before failure. The data is given in Table \ref{tab:Application_data}. It can be seen from the histogram in Figure \ref{fig:GE_Bayes_app_fit_zoom} that the distribution is right skewed.

\begin{table}[htbp!] \label{tab:Application_data}
\centering
\small
\caption{A real data set from \citet{lawless_statistical_1982} }
\begin{tabular}{rrrrrrrrrrrrr} \hline
  17.88 & 28.92 & 33.00 & 41.52 & 42.12 & 45.60 & 48.40  & 51.84  & 51.96  & 54.12  & 55.56  & 67.80 \\
  68.64 & 68.64 & 68.88 & 84.12 & 93.12 & 98.64 & 105.12 & 105.84 & 127.92 & 128.04 & 173.40 & \\
  \hline
\end{tabular}
\end{table}

The parameters are estimated using the independence Jeffreys prior $(a=1, b=1)$ and model diagnostics are included in Figures \ref{fig:GE_Bayes_app_trace_density} and \ref{fig:GE_Bayes_app_acf}. The trace plots and density plots of the marginal posterior samples suggest that the chain is mixing well and has reached convergence. This observation is supported by the z-scores, $-$0.8674 and $-$1.493 for $\alpha$ and $\lambda$, from the Geweke's convergence diagnostic (i.e., $|Z| < 1.96$). The auto-correlation function (ACF) plots show no correlation between the posterior samples for $\alpha$ and $\lambda$.

\begin{figure}[htbp!]
\centering
\caption{Trace and density plots of the marginal posterior samples}
\label{fig:GE_Bayes_app_trace_density}
\includegraphics[height=10cm, width=1\linewidth]{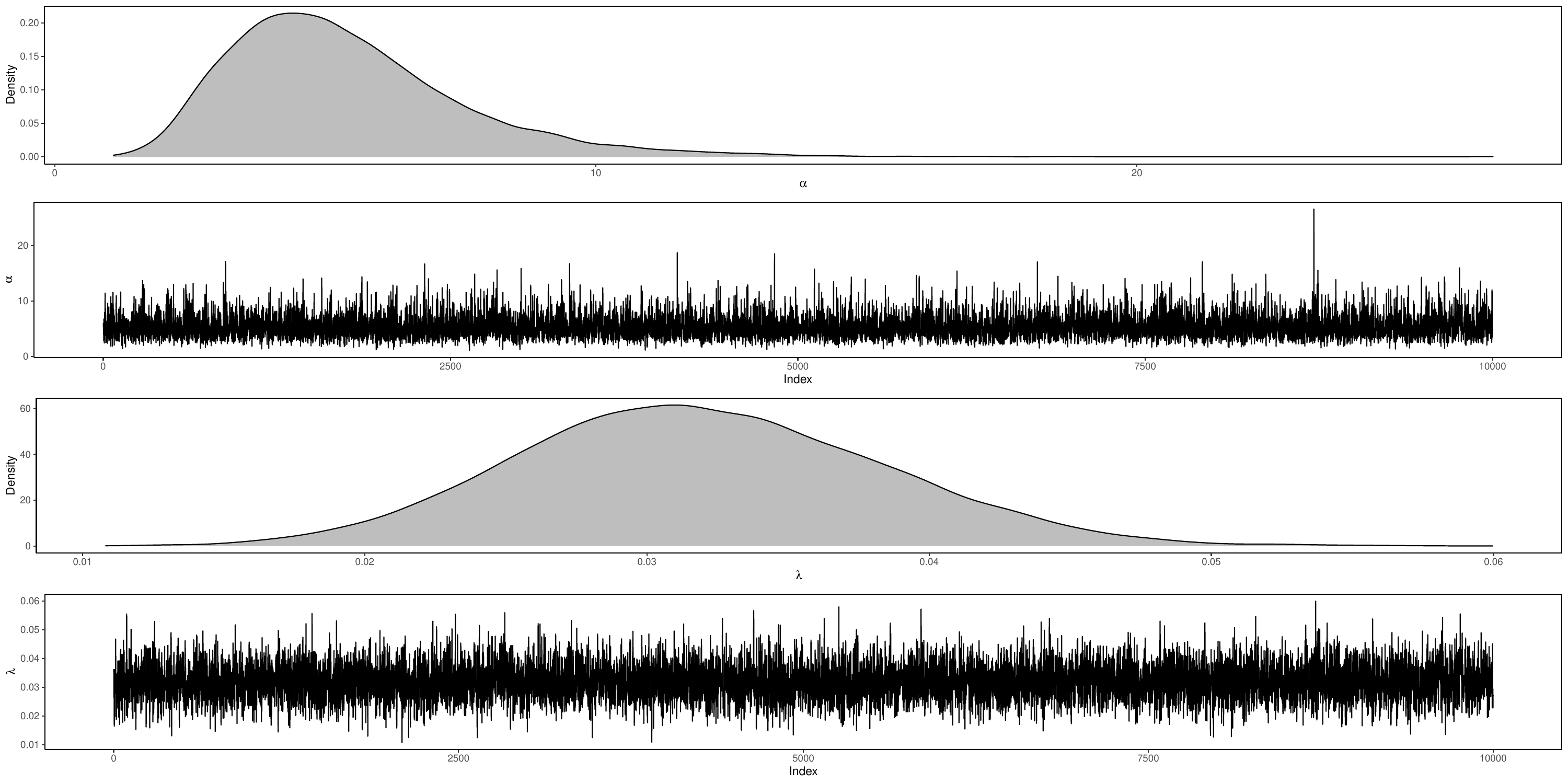}
\end{figure}

\begin{figure}[ht!]
\centering
\caption{ACF plots of the marginal posterior samples}
\label{fig:GE_Bayes_app_acf}
\includegraphics[width=1\linewidth]{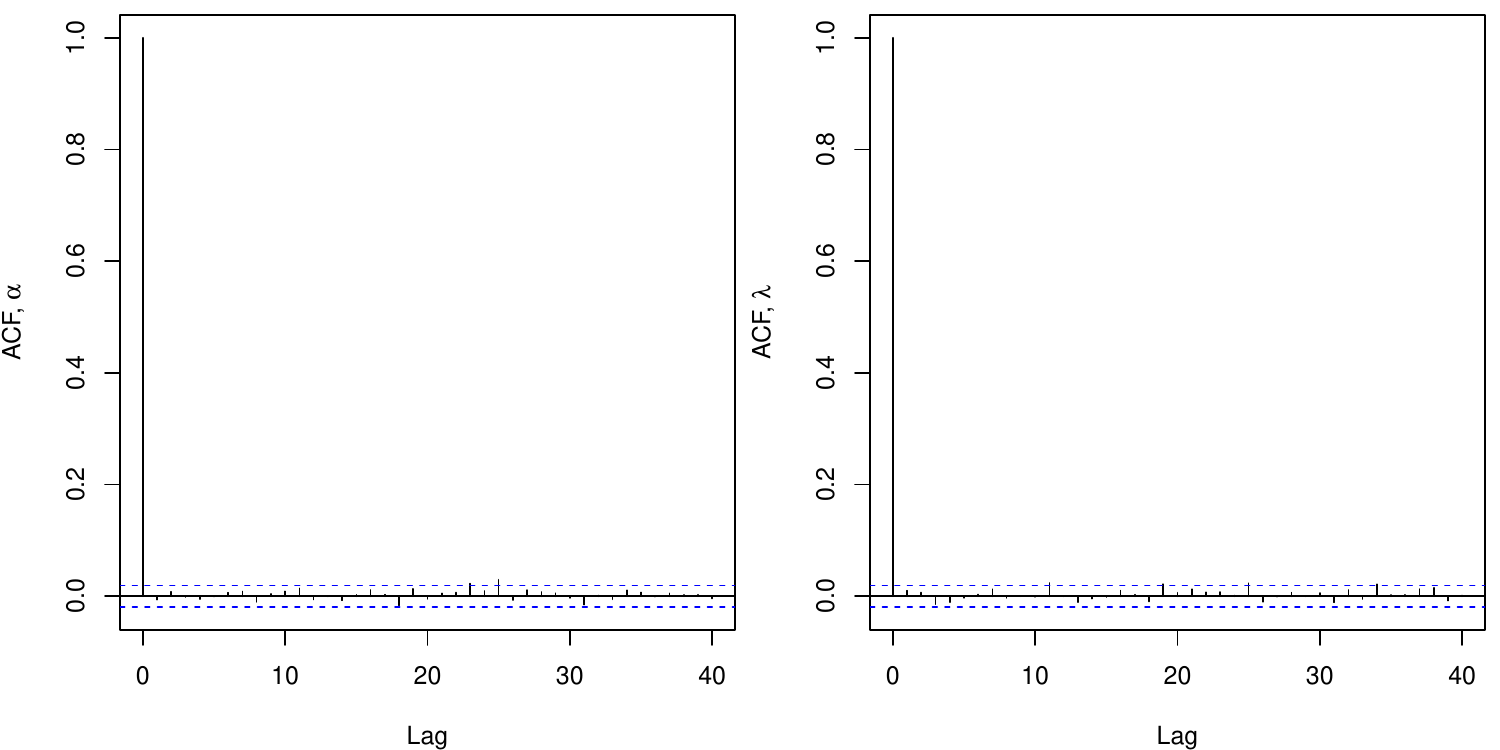}
\end{figure}

Table \ref{tab:Bayes_est} shows the Bayes and MLE estimates of the parameters $\alpha$ and $\lambda$. Kolmogorov-Smirnov (K-S) statistics and p-value are calculated to compare the empirical distribution and the fitted GE distribution with the estimated parameters. Furthermore, a histogram of the data with the fitted densities (Figure \ref{fig:GE_Bayes_app_fit_zoom}) are included to visualize the fitness of the models. Based on the results, the GE distribution estimated using the independence Jeffreys prior is a good fit for the empirical distribution.

\begin{table}[!htbp] \label{tab:Bayes_est}
%\small\addtolength{\tabcolsep}{-0.3pt}
\centering % centering table
\caption{Estimation of the GE parameters for the data from \citet{lawless_statistical_1982}.}
\begin{tabular}{rrrrr}    \hline
Method                  & $\hat{\alpha}$ & $\hat{\lambda}$ &  K-S     & p-value \\ \hline
Bayes                   & 5.0219         & 0.0317          & 0.10432  & 0.9416   \\
MLE                     & 5.2783         & 0.0322          & 0.10588  & 0.9349      \\
\hline
\end{tabular}
\end{table}

\begin{figure}[htbp!]
\centering
\caption{Histogram with fitted densities using the Bayes (solid line) and MLE (dashed line) estimators}
\label{fig:GE_Bayes_app_fit_zoom}
\includegraphics[width=1\linewidth]{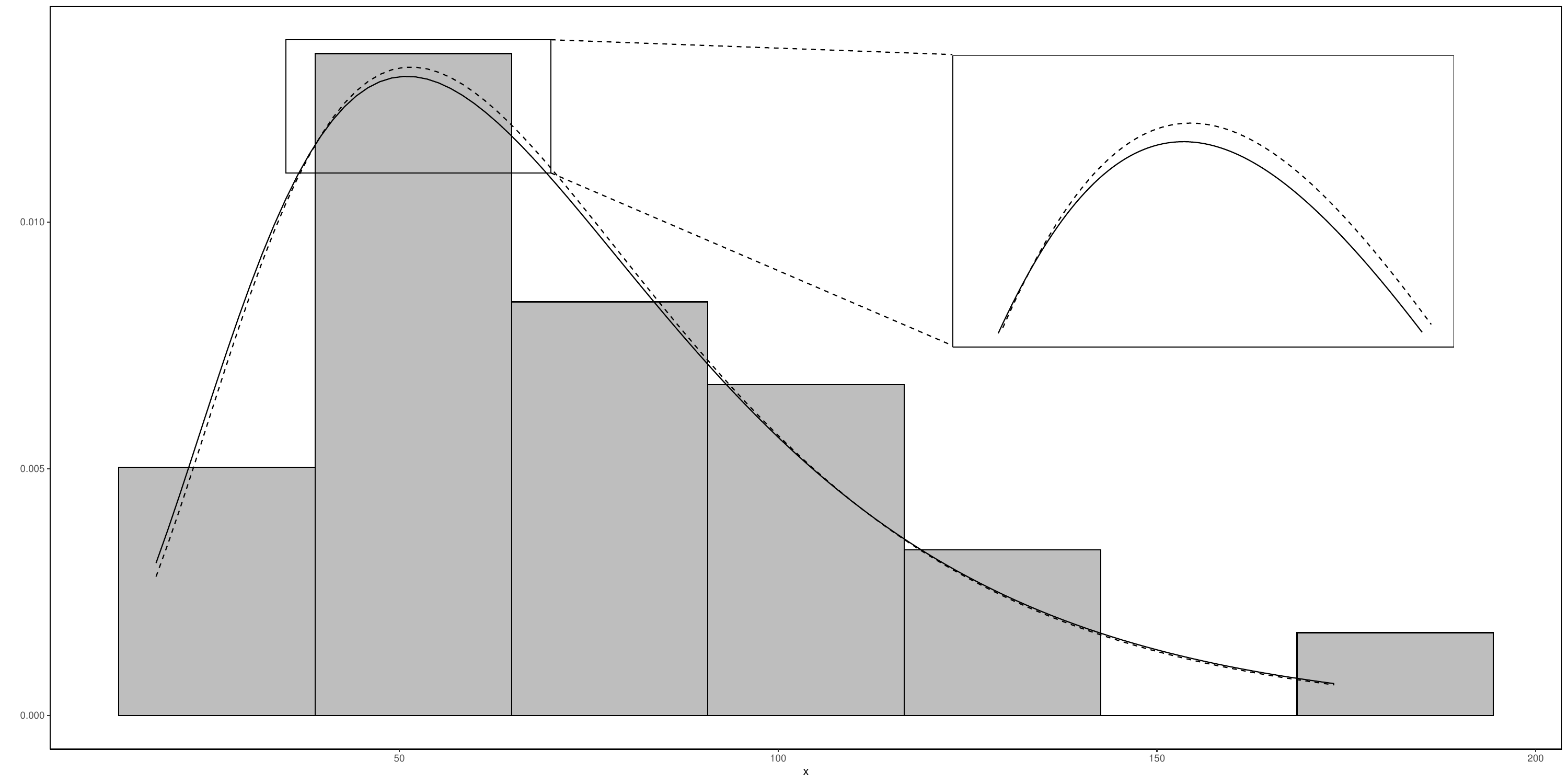}
\end{figure}

\section{Concluding remarks}
\label{sec:con}
In this paper, we have studied Bayesian inference for the GE distribution under a family of vague priors in (\ref{eq:pp}) and validated the propriety for the posterior distribution of the parameters  $a \ge 1$ and $b \le 1$. We have also proposed a generalized ratio-of-uniforms method for generating posterior samples for making statistical inference. It is worth mentioning that the proposed algorithm not only generates independent samples from the posterior distribution, but also is independent of starting values. Numerical results from both simulations and a real-data application showed that the Bayesian approach outperformed the MLE in terms of SBias and SRMSE, especially when the sample size is small or moderate. In summary, we have a preference for the Bayes estimators as an alternative to the MLE for the parameters of the GE distribution in practical applications. %Finally, a real data analysis further confirmed the suprior performance of the Bayesian method in terms of the KS tests. %that it is feasible to use objective Bayesian analysis for the GE distribution in practice.

%For further research, one can try to validate the propriety of other non-informative priors, including probability matching priors, the maximal data information prior, and reference priors.

\section{Appendix}

\setcounter{section}{1}
\setcounter{equation}{0}
\setcounter{figure}{0}
\setcounter{table}{0}
\def\theequation{A\arabic{equation}}
\def\thesection{A\arabic{section}}

\renewcommand{\thefigure}{A\arabic{figure}}
\renewcommand{\thetable}{A\arabic{table}}

\textbf{Proof of Theorem \ref{theorem_proper}}: We first introduce the following two important lemmas that play key roles in  validating the propriety of Equation \eqref{eq:mpost_lambda}.
\begin{lemma}
\label{lemma_bound}
If $x>0$, then $x-\frac{1}{2}x^2<1-e^{-x}<x$.
\end{lemma}
\begin{proof}
Let $g(x) = 1 - e^{-x}-x +\frac{1}{2}x^2$. It can be easily shown that \begin{eqnarray}\label{lemma_bound_der}
g'(x) = e^{-x}-1+x,\text{ and }g''(x) = 1- e^{-x}.\end{eqnarray}
From the second equation in (\ref{lemma_bound_der}), we see that $g''(x)>0$ for $x>0$. Hence $g'(x)$ is an increasing function in $x$, and $e^{-x}-1+x = g'(x) > g'(0) = 0$ for $x>0$. On the other hand, $g'(x)>0$ for $x>0$ implies that $g(x)> g(0) = 0$ for $x>0$. We thus proved the lemma.
\end{proof}

\begin{lemma}
\label{lemma_infty}
For a positive integer $n\ge 2$, suppose $a \ge 1$, $n-a+1>0$, $x_i>0$ for $i=1,\dots n$, and not all $x_i$'s are equal. There exists at least one $i$, such that $$\frac{\sum_{j=1}^{n}x_j - \delta}{n-a+1} - x_i > 0,$$ where $\delta > 0$ is a constant only depends on $\bm X = (x_1,\dots,x_n)'$.
\end{lemma}

\begin{proof}
Since not all $x_i$'s are equal, we know there exists at least one $i$ such that $\bar{x}-x_i>0$. Now $$\frac{\sum_{j=1}^{n}x_j}{n-a+1} =\frac{n}{n-a+1} \bar{x} >\bar{x},\text{ due to }a\ge 1.$$ Hence, there exists an $\eta>0$ such that $$\frac{\sum_{j=1}^{n}x_j}{n-a+1} - x_i>\eta.$$ Let $\delta = \eta(n-a+1)$. The result is immediate.
\end{proof}

We now provide the proof of Theorem \ref{theorem_proper} as follows.
\begin{proof}
To prove the propriety of the marginal posterior distribution of $\lambda$ in \eqref{eq:mpost_lambda}, we only need to prove that the density is integrable both when either $\lambda\rightarrow 0$ or $\lambda\rightarrow\infty$.\\
First let's consider the case when $\lambda\rightarrow\infty$. Equation \eqref{eq:mpost_lambda} can be written as
\begin{align}
    \pi(\lambda \mid \bm X)
    =  \lambda^{n-b} e^{-\lambda \delta} \prod_{i=1}^{n}\left(1-e^{-\lambda x_i}\right)^{-1} \Gamma(n-a+1) \left\{\frac{e^{-\lambda \frac{\sum_{j=1}^{n}x_j - \delta}{n-a+1}} }{-\sum_{i=1}^{n}\ln\left(1-e^{-\lambda x_i}\right)}\right\}^{n-a+1}. \label{eq:mpost_lambda2}
\end{align}
The first two terms in \eqref{eq:mpost_lambda2} is a gamma density kernel and hence they are integrable at $\lambda\rightarrow\infty$. The third and fourth terms are bounded near infinity. Hence we only need to check the ratio inside the brackets of the last term.  It can be easily seen that both the numerator and denominator go to zeroes when $\lambda\rightarrow\infty$. Hence we use L'Hospital's rule to derive the following
\begin{align*}
\underset{\lambda\rightarrow\infty}{\lim} &\left\{\frac{e^{-\lambda \frac{\sum_{j=1}^{n}x_j - \delta}{n-a+1}}}{-\sum_{i=1}^{n}\ln\left(1-e^{-\lambda x_i}\right)}\right\} = \underset{\lambda\rightarrow\infty}{\lim} \left\{\frac{- \frac{\sum_{j=1}^{n}x_j - \delta}{n-a+1}e^{-\lambda \frac{\sum_{j=1}^{n}x_j - \delta}{n-a+1}}}{-\sum\limits_{i=1}^n\frac{x_i e^{-\lambda x_i}}{1- e^{-\lambda x_i}} }\right\} \\
&= \frac{\sum_{j=1}^{n}x_j - \delta}{n-a+1} \underset{\lambda\rightarrow\infty}{\lim} \left\{\sum\limits_{i=1}^n \frac{x_i}{1-e^{-\lambda x_i}}e^{\lambda \left(\frac{\sum_{j=1}^{n}x_j - \delta}{n-a+1} -x_i \right)} \right\}^{-1}.
\end{align*}
Using Lemma \ref{lemma_infty} we know that there exists at least one $i$ such that $\frac{\sum_{j=1}^{n}x_j - \delta}{n-a+1} -x_i>0$ for the exponential terms in the last line of the above equation, which implies that the limit tends to 0 as $\lambda\rightarrow \infty$.\\
To show that the marginal posterior distribution of $\lambda$ is integrable near 0, we consider the following. Using Lemma \ref{lemma_bound}, we know $x-\frac{1}{2}x^2<1-e^{-x}<x$ for any $x>0$. Therefore, it follows that
$$
-\ln\left(1- e^{-\lambda x_i} \right) > \ln(\lambda x_i),\text{ and }1 - e^{-\lambda x_i} > \lambda\left(x_i - \frac{\lambda x_i^2}{2}\right).
$$
Now applying these inequalities into \eqref{eq:mpost_lambda}, we have
\begin{align}\label{Upper_bounds}
\pi(\lambda \mid \bm X) &< \lambda^{-b}e^{- \lambda\sum\limits_{i=1}^ n x_i}\prod\limits_{i=1}^n\left(x_i - \frac{\lambda x_i^2}{2}\right)^{-1} \frac{ \Gamma(n-a+1)}{\left\{-n\ln(\lambda)- \sum\limits_{i=1}^n\ln(x_i) \right\}^{n-a+1}}.
\end{align}
All terms, except $\lambda^{-b}$, in (\ref{Upper_bounds}) are bounded when $\lambda\rightarrow 0$. Note that the denominator of the last ratio approaches infinity. Hence, when $b<1$, the right-hand side of (\ref{Upper_bounds}) is integrable at $\lambda$ near 0. When $b=1$, we consider the first term and the denominator term in (\ref{Upper_bounds}). Since $$\int\lambda^{-1}\frac{1}{\left\{-n\ln(\lambda)- \sum\limits_{i=1}^n\ln(x_i) \right\}^{n-a+1} }d\lambda = \frac{1}{n(n-a)}\left\{-n\ln(\lambda)- \sum\limits_{i=1}^n\ln(x_i) \right\}^{-(n-a)} + C.$$ When $n> a$, the above integral is integrable at $\lambda\rightarrow 0$.
\end{proof}

    \bibliographystyle{ajs}

\end{document}